\documentclass[11pt,reqno]{amsart}
\usepackage[T1]{fontenc}
\usepackage{lmodern}
\usepackage{amsmath,amssymb,amsthm,mathtools}
\usepackage[margin=1.08in]{geometry}
\usepackage{microtype}
\usepackage{enumitem}
\usepackage{needspace}
\usepackage{xcolor}
\usepackage{comment}

\usepackage{algorithm}
\usepackage{algpseudocode}
\algrenewcommand{\algorithmicrequire}{\textbf{Input:}}
\algrenewcommand{\algorithmicensure}{\textbf{Output:}}

\usepackage[colorlinks=true,linkcolor=blue!45!black,citecolor=blue!45!black,urlcolor=blue!45!black]{hyperref}
\hypersetup{pdftitle={A deterministic algorithm for signing bipartite graphs at the Ramanujan bound},pdfauthor={}}
\newtheorem{theorem}{Theorem}[section]
\newtheorem{lemma}[theorem]{Lemma}
\newtheorem{proposition}[theorem]{Proposition}
\newcommand{\R}{\mathbb R}
\newcommand{\norm}[1]{\left\lVert#1\right\rVert}
\newcommand{\Fill}{\textsc{Fill}}
\newtheorem{definition}[theorem]{Definition}

\newtheorem{question}[theorem]{Question}
\numberwithin{equation}{section}
\setlist[enumerate]{leftmargin=*,itemsep=4pt,topsep=6pt}
\title[Deterministic bipartite signings]{A deterministic polynomial-time algorithm for the
Bilu--Linial signing problem on bipartite graphs}
\author{Zhiqiang Xu}
\thanks{Zhiqiang Xu is supported by National Natural Science Foundation of China (Grant No. 12471361, 12288201).}
\address{State Key Laboratory of Mathematical Sciences, Academy of Mathematics and Systems Science, Chinese Academy of Sciences, Beijing 100190, China;    School of Mathematical Sciences, University of Chinese Academy of Sciences, Beijing 100049, China. }
\email{xuzq@lsec.cc.ac.cn}
\date{}

\begin{document}
\begin{abstract}
We give a deterministic polynomial-time algorithm for the
Bilu--Linial signing problem on bipartite graphs.
For every finite simple bipartite graph of maximum degree
at most an integer $\Delta\ge3$, the algorithm assigns signs
to its edges so that the signed adjacency matrix has operator
norm strictly less than $2\sqrt{\Delta-1}$.
Our algorithm builds on the randomized recursive repair
framework of Jadbabaie, Saberi, and Sra~\cite{JSS26},
with deterministic rules for sign selection and vertex deletion.
\end{abstract}
\maketitle

\section{Introduction}

\subsection{The signing problem}
All graphs in this paper are finite, undirected, and simple unless
explicitly stated otherwise. For a graph $G=(V,E)$, where $V$ is its
vertex set and $E$ its edge set, a \emph{signing} is a map
$\sigma:E\to\{-1,1\}$. The associated signed adjacency matrix
$\mathbf{A}_\sigma$ is indexed by $V$ and has entries
\[
 (\mathbf{A}_\sigma)_{uv}=
 \begin{cases}
   \sigma(\{u,v\}),&\{u,v\}\in E,\\
   0,&\{u,v\}\notin E.
 \end{cases}
\]
We write $\norm{\cdot}$ for the operator norm induced by the Euclidean
norm. For a real symmetric matrix, this is the largest absolute value
of an eigenvalue. Matrices and vectors are denoted by bold letters;
their indices and scalar coordinates are not bold. The superscript
${}^{\mathsf T}$ denotes transpose.

Bilu and Linial~\cite{BL06} connected graph signings with $2$-lifts.
A $2$-lift replaces each vertex by two copies. A positive edge joins
corresponding copies of its endpoints, and a negative edge joins
opposite copies. The lift's eigenvalues consist, with multiplicity,
of those of the original adjacency matrix and those of
$\mathbf{A}_\sigma$~\cite[Lemma~3.1]{BL06}. Thus choosing the edge
signs controls the new eigenvalues of the lift.
Bilu and Linial~\cite[Conjecture~3.1]{BL06} conjectured that every
$d$-regular graph, with $d\ge2$, admits an edge signing $\sigma$ such that
\(
\|\mathbf{A}_\sigma\|\le 2\sqrt{d-1}.
\)

Marcus, Spielman, and Srivastava~\cite{MSS15} proved this conjecture
for bipartite graphs using interlacing families.
More generally, their argument shows that every nonempty graph
admits a signing whose largest eigenvalue is at most the largest
root of its matching polynomial; see the proof
of~\cite[Theorem~5.3]{MSS15}.
Together with the Heilmann--Lieb bound recalled
in~\cite[Theorem~3.2]{MSS15}, this yields an upper bound
$2\sqrt{\Delta-1}$ when the maximum degree is at most an integer
$\Delta\ge3$.
For bipartite graphs, the signed spectrum is symmetric about zero,
so the same signing satisfies
\(
\|\mathbf{A}_\sigma\|\le 2\sqrt{\Delta-1}.
\)

Consequently, every $d$-regular bipartite Ramanujan graph admits
a $2$-lift that is again a $d$-regular bipartite Ramanujan graph.
Starting from the complete bipartite graph $K_{d,d}$ and iterating
such lifts therefore gives an infinite family of $d$-regular
bipartite Ramanujan graphs for every $d\ge3$,
as established in~\cite[Theorem~5.5]{MSS15}.
To implement this construction efficiently, one needs an efficient
method for finding the required signings.
This motivates the following computational question.

\begin{question}\label{que:1}
Given a finite simple bipartite graph of maximum degree at most an
integer $\Delta\ge3$, can one find in deterministic polynomial time
an edge signing whose signed adjacency matrix has operator norm
at most $2\sqrt{\Delta-1}$?
\end{question}

\Needspace{20\baselineskip}
\subsection{Main result}
We answer Question~\ref{que:1} affirmatively with a deterministic
polynomial-time algorithm for signing the prescribed bipartite
graph. The algorithm is described in Algorithm~\ref{alg:signing}.
We measure running time in bit operations, with integers
represented in binary. The following theorem states the strict
spectral bound and the running-time guarantee.
\begin{theorem}\label{thm:main}
Let $\Delta\ge3$ be an integer, and let $G=(V,E)$ be a finite
simple bipartite graph with $V=\{1,\ldots,n\}$, $n\ge2$,
and maximum degree at most $\Delta$.
Algorithm~\ref{alg:signing} deterministically computes a signing
$\sigma:E\to\{-1,1\}$ whose signed adjacency matrix
$\mathbf{A}_\sigma$ satisfies
\begin{equation}\label{eq:main}
 \norm{\mathbf{A}_\sigma}<2\sqrt{\Delta-1}.
\end{equation}
The algorithm can be implemented using exact rational arithmetic
with
\[
 O\!\left(n^{10}
   \bigl(\log(n+1)+\log(\Delta+1)\bigr)^2\right)
\]
bit operations, assuming that $\Delta$ is given in binary and
$G$ is represented by an explicit vertex list and edge list,
or by an adjacency matrix.
\end{theorem}

The graph $G$ need not be regular or connected.
The stated bit-complexity bound is obtained from a direct
implementation that recomputes the required matrix data whenever
they are needed. We do not claim that the exponent $10$ is optimal.

There exist nonbipartite graphs for which no edge signing yields
a signed adjacency matrix with operator norm at most
$2\sqrt{\Delta-1}$.
Specifically, in earlier work~\cite[Theorem~1.1]{Xu26}, the author
constructed a finite connected simple $3$-regular graph such that
every edge signing produces a signed adjacency matrix with operator
norm strictly greater than $2\sqrt{2}$.
Thus the bipartite assumption in Theorem~\ref{thm:main}
cannot be omitted.

The proof of Theorem~\ref{thm:main} does not rely on the results
of Marcus, Spielman, and Srivastava in
\cite{MSS15}.
It therefore gives an alternative constructive proof of the
Bilu--Linial signing conjecture for bipartite graphs.

\subsection{Related work}

Marcus, Spielman, and Srivastava~\cite[Theorem~1.1]{MSS18}
proved that, for every even integer $n\ge2$ and every integer
$\Delta\ge3$, there exists a $\Delta$-regular bipartite
Ramanujan multigraph on $n$ vertices.
Their proof uses interlacing families and finite free
convolutions to analyze unions of random perfect matchings.
Building on this framework, Cohen~\cite{Cohen16} obtained
a deterministic polynomial-time construction by efficiently
computing the conditional expected characteristic polynomials
needed to select the matchings.
These results concern graphs whose edges are chosen during
the construction.
Our algorithm instead takes a prescribed simple bipartite
graph $G=(V,E)$ of maximum degree at most $\Delta$ and
deterministically assigns signs to its existing edges in
polynomial time, ensuring
$\norm{\mathbf{A}_\sigma}<2\sqrt{\Delta-1}$ without requiring
regularity or connectedness.

Jadbabaie, Saberi, and Sra~\cite[Theorem~6.1]{JSS26} give a
Las Vegas algorithm for signing a prescribed graph.
On bipartite inputs, their randomized algorithm terminates
almost surely and always returns a signing
satisfying~\eqref{eq:main}. It admits an implementation using
exact rational arithmetic with expected polynomial bit
complexity~\cite[Section~6.6]{JSS26}.
They also suggest conditional expectations as a possible route
to derandomization, but do not pursue
it~\cite[Section~7]{JSS26}.

Our algorithm builds on their recursive repair framework and
local determinant identities, and adapts the finite recurrence
used in their analysis~\cite[Section~6.5]{JSS26}.
By modifying both the sign-selection and deletion rules,
we obtain a fully deterministic algorithm that achieves the
same spectral bound for prescribed bipartite graphs.
The resulting bit-complexity guarantee is polynomial in the
worst case, rather than in expectation.

\subsection{Organization}
Section~\ref{sec:ingredients} introduces the notation and states
the auxiliary results needed to describe and analyze the algorithm.
Section~\ref{sec:algorithm} presents the algorithm and proves
Theorem~\ref{thm:main}.
Section~\ref{sec:details} proves the auxiliary results stated
in Section~\ref{sec:ingredients}, including those concerning
the exact implementation of the local computations and the
associated bit-complexity bounds.
Section~\ref{sec:conc} discusses the complexity in the
unit-cost arithmetic model and possible directions for future research.

\section{Notation and lemmas for the algorithm and its analysis}\label{sec:ingredients}
For the remainder of the proof, fix an input $G=(V,E)$ and $\Delta$
as in Theorem~\ref{thm:main}, with $n=|V|\ge2$. Vertex labels fix all
orders and tie-breaking choices. Set
\begin{equation}\label{eq:parameters}
 \rho_0:=2\sqrt{\Delta-1},\qquad
 \gamma_0:=1+\frac{1}{64n^2}.
\end{equation}
The first constant is the spectral threshold and the second will
measure progress at each insertion or deletion.

\subsection{The active state}
\begin{definition}[Active states and admissibility]
An \emph{active state} is a pair $(K,\sigma)$, where $K\subseteq V$
and $\sigma$ assigns one sign to every edge of the induced graph
$G[K]$. We call $K$ the active vertex set.
Thus only edges with both endpoints in $K$ currently carry
signs. Let $\mathbf{A}_K$ be the signed adjacency matrix of this
state and define
\begin{equation}\label{eq:state}
 \mathbf{Q}_K=\mathbf{I}_K-\frac{\mathbf{A}_K}{\rho_0},
 \qquad D_K=\det\mathbf{Q}_K,
\end{equation}
where $\mathbf{I}_K$ is the identity matrix indexed by $K$.
The state $(K,\sigma)$ is \emph{admissible} if $\mathbf{Q}_K\succ 0$.
For an admissible state and $u\in K$, put
\begin{equation}\label{eq:leverage}
 \lambda_u=(\mathbf{Q}_K^{-1})_{uu}.
\end{equation}
The empty state is admissible. Its determinant is $1$, and its
inverse is the empty matrix.

\end{definition}
All these quantities depend on the current signing $\sigma$,
although only $K$ appears in the notation.
After a recursive call reinserts a deleted vertex, the vertex
set $K$ is the same as before the deletion, but $\sigma$ may
have changed. Thus the matrices $\mathbf{A}_K$ and $\mathbf{Q}_K$,
as well as the determinant $D_K$, may differ from their values
before the deletion.
For bipartite states, admissibility is equivalent to
\begin{equation}\label{eq:spectralconstraint}
 \norm{\mathbf{A}_K}<\rho_0.
\end{equation}
For $K=\varnothing$, we take $\norm{\mathbf{A}_K}=0$.

\subsection{Deleting a vertex and choosing insertion signs}
Recall that $K\subseteq V$ denotes the current set of active
vertices, and $\sigma:E(G[K])\to\{-1,1\}$ denotes the current
edge signing on the induced subgraph $G[K]$.
Deleting an active vertex $u\in K$ removes $u$ from $K$ and
erases its incident signs, while retaining all signs on
$G[K\setminus\{u\}]$.
Inserting an inactive vertex $v\in V\setminus K$ adds $v$
to $K$, retains all existing signs, and assigns signs to
the edges joining $v$ to its neighbors already in $K$.
Both operations modify only the temporary state $(K,\sigma)$ and 
the input graph $G$ remains fixed.
The next two lemmas quantify the effects of these operations.

For an admissible state $(K,\sigma)$ and a vertex $v\notin K$,
write $U_v=\{u\in K:\{u,v\}\in E\}$ for its active neighbors.
The dependence of $U_v$ on the current set $K$ is suppressed
in the notation.
A vector $\mathbf{s}=(s_u)_{u\in U_v}\in\{-1,1\}^{U_v}$
specifies the signs on the edges joining $v$ to $U_v$.
When inserting $v$ into $K$, we extend $\sigma$ by setting
$\sigma(\{v,u\})=s_u$ for every $u\in U_v$, while retaining
all signs on $G[K]$.
 Define the insertion quadratic form by
\begin{equation}\label{eq:qdef}
 q_v(\mathbf{s})=
 \frac{1}{\rho_0^2}\cdot \mathbf{s}^{\mathsf T}
      (\mathbf{Q}_K^{-1})[U_v,U_v]\mathbf{s}.
\end{equation}
Here $[U_v,U_v]$ denotes restriction to the rows and columns indexed
by $U_v$, in label order. The form is zero when $U_v$ is empty.

\begin{lemma}\label{lem:determinants}
For every admissible state, $0<D_K\le1$. Deleting $u\in K$
preserves admissibility and gives
\begin{equation}\label{eq:delete}
 D_{K\setminus\{u\}}=D_K\cdot \lambda_u.
\end{equation}
Here, \(D_K\) and \(\lambda_u\) are defined in (\ref{eq:state}) and  (\ref{eq:leverage}), respectively. For $v\notin K$, inserting $v$ with incident sign vector
$\mathbf{s}$ is admissible if and only if $q_v(\mathbf{s})<1$.
In that case,
\begin{equation}\label{eq:insert}
 D_{K\cup\{v\}}=D_K\cdot \bigl(1-q_v(\mathbf{s})\bigr).
\end{equation}
The determinants on the left refer to the states obtained by the
specified deletion or insertion, respectively.
\end{lemma}
\begin{lemma}[Greedy sign selection]\label{lem:signs}
Let $(K,\sigma)$ be admissible and let $v\in V\setminus K$.
List the neighbors of $v$ currently in $K$ as
$U_v=\{u_1<\cdots<u_\ell\}$, where $\ell=|U_v|$.
Choose the coordinates of the sign vector
$\mathbf{s}=(s_u)_{u\in U_v}$ successively by
\begin{equation}\label{eq:greedy}
 s_{u_i}=
 \begin{cases}
 -1,&\displaystyle
       \sum_{k=1}^{i-1}(\mathbf{Q}_K^{-1})_{u_i u_k}s_{u_k}>0,\\
 +1,&\text{otherwise},
 \end{cases}
 \qquad 1\le i\le\ell.
\end{equation}
Then
\begin{equation}\label{eq:qtrace}
 q_v(\mathbf{s})\le
 \frac{1}{\rho_0^2}\sum_{u\in U_v}\lambda_u.
\end{equation}
Given the inverse entries appearing in~\eqref{eq:greedy},
this rule uses $O(\ell^2)$ arithmetic operations and comparisons.
Empty sums are zero and if $U_v=\varnothing$, no signs are chosen.
\end{lemma}

Recall that each coordinate $s_u$ specifies the sign of the edge $\{v,u\}$.
Thus, when inserting $v$ into $K$, we extend the current signing
by setting $\sigma(\{v,u\})=s_u$ for every $u\in U_v$,
while retaining all signs on $G[K]$.
The key point is that the inverse diagonal entries $\lambda_u$
control the determinant changes under both insertion and deletion.

\subsection{Weights for the algorithm}

We introduce positive weights $C_1,\ldots,C_n$ to determine
when the algorithm deletes a vertex from $K$ and when it
inserts a vertex into $K$.
These weights specify the deletion thresholds and the bounds
on the determinant ratios needed to prove termination.
The algorithm in Section~\ref{sec:algorithm} explains how
the appropriate weight is selected at each step.
The following lemma ensures that all weights are well defined
and can be computed exactly in polynomial time.

\begin{lemma}[Depth weights]\label{lem:weights}
Define $C_1,\ldots,C_n$ backward by
\begin{align}
 C_n&=\gamma_0,\label{eq:Cterminal}\\
 C_{j}&=\frac{\gamma_0}{1-\gamma_0C_{j+1}/4}
       &&(2\le j<n),\label{eq:Cinternal}\\
 C_1&=\frac{\gamma_0}{1-\Delta\gamma_0C_2/\rho_0^2}.
       \label{eq:Croot}
\end{align}
All denominators are positive, and
\begin{equation}\label{eq:Cbound}
 \gamma_0\le C_{j}<2\gamma_0\quad(2\le j\le n),
 \qquad \gamma_0\le C_1<5\gamma_0.
\end{equation}
Each weight can be represented as a ratio of positive integers,
each having
$O(n(\log(n+1)+\log(\Delta+1)))$ bits.
All these representations can be computed in
\(
 O\!\left(n^3
   \bigl(\log(n+1)+\log(\Delta+1)\bigr)^2\right)
\)
bit operations.
\end{lemma}

For $n=2$, the middle recurrence is empty.
The weights depend only on $n$ and $\Delta$ and remain fixed
throughout the computation.

\subsection{The cost of a local decision}

The deletion test and the greedy rule require only the entries
$(\mathbf{Q}_K^{-1})_{uw}$ with $u,w\in U_v$.
Although other entries of $\mathbf{Q}_K^{-1}$ may be irrational,
these required entries are rational.
Indeed, all vertices in $U_v$ belong to the bipartition class
opposite to $v$. As shown in the proof below, the corresponding
block of $\mathbf{Q}_K^{-1}$ can be computed using only
$\rho_0^2=4(\Delta-1)$ and the current signs.
The following proposition bounds the bit cost of these
computations and the resulting sign choices and comparisons.

\begin{proposition}[Exact local computation]\label{prop:arithmetic}
Let $(K,\sigma)$ be an admissible state of the bipartite graph $G$,
let $v\in V\setminus K$, and let $1\le \jmath_0\le n$.
Write $N_G(v)=\{u\in V:\{u,v\}\in E\}$ for the neighbors of $v$
in $G$. Recall that $U_v=N_G(v)\cap K$.
The entries $(\mathbf{Q}_K^{-1})_{uw}$ with $u,w\in U_v$
are rational.
One can compute these entries, test
$\lambda_u\ge\gamma_0C_{\jmath_0+1}$ for every $u\in U_v$
when $\jmath_0<n$, and, when needed, choose the signs by
\eqref{eq:greedy}, using exact rational arithmetic in
\[
 O\!\left(n^7
   \bigl(\log(n+1)+\log(\Delta+1)\bigr)^2\right)
\]
bit operations.
This bound includes forming the required matrix data from
the current signing $\sigma$ and is independent of the
preceding execution history.
\end{proposition}

Thus the local decisions can be implemented exactly without
forming the possibly irrational entries of $\mathbf{Q}_K$
or the entire inverse $\mathbf{Q}_K^{-1}$.
The implementation uses only the integer $\rho_0^2$,
the current signs, and the rational weights
$C_1,\ldots,C_n$; no approximation of $\rho_0$ or numerical
eigenvalue test is needed.
Combining this local bit bound with the bound on the number
of updates gives the total bit-complexity bound in
Theorem~\ref{thm:main}.

\section{The algorithm and its analysis}\label{sec:algorithm}

In this section, we present the deterministic signing algorithm 
and prove Theorem~\ref{thm:main} using the results stated in
Section~\ref{sec:ingredients}.

\subsection{The signing algorithm}
\subsubsection{The main routine and its interface}

Algorithm~\ref{alg:signing} starts from the empty active state
$(K,\sigma)=(\varnothing,\varnothing)$ (Lines~12--13, Algorithm~\ref{alg:signing}) and processes vertices
in label order. It uses the subroutine $\Fill$, specified in
Algorithm~\ref{alg:fill}, to add an inactive vertex
$v\in V\setminus K$ to $K$ while maintaining the admissibility
of the state $(K,\sigma)$, allowing temporary deletions and
changes of signs along the way.
Algorithm~\ref{alg:fill} is used only for calls generated
during the execution of Algorithm~\ref{alg:signing},
including recursive calls. Its output guarantee is
asserted for these calls.
The subroutine receives the current pair $(K,\sigma)$ and
returns an updated pair.

An \emph{outer call} (Line~15, Algorithm~\ref{alg:signing}) is a call to $\Fill$ made directly by
Algorithm~\ref{alg:signing}. A \emph{child call} (Line~15, Algorithm~\ref{alg:fill}) is a recursive
call made by $\Fill$ to restore a deleted neighbor $u_*$;
the calling procedure waits until the child call returns.

In each call to $\Fill$, $v$ denotes its fourth input argument,
the vertex to be inserted. The input $P$ records the vertices
of all unfinished ancestor calls. These vertices are inactive
and distinct from $v$, so $P\subseteq V\setminus(K\cup\{v\})$.
The recursion depth is $\jmath_0=|P|+1$
(Line~2, Algorithm~\ref{alg:fill}), which determines
the deletion threshold.

An outer call receives $P=\varnothing$ and the current state
$(K,\sigma)$. If a call for $v$ deletes a neighbor $u_*$,
it invokes a child call with $u_*$ as its fourth argument
and $P\cup\{v\}$ as its ancestor set.
Upon returning, the child call has reinserted $u_*$
and restored the active vertex set to its value before
the deletion. The signing $\sigma$, however, may have changed,
since reinsertion assigns new signs to the incident edges.
The parent's vertex $v$ remains inactive throughout
the child call.

\begin{algorithm}[H]
\caption{Deterministic signing of a bipartite graph}
\label{alg:signing}
\begin{algorithmic}[1]
\Require A finite simple bipartite graph $G=(V,E)$ with
vertex set $V=\{1,\ldots,n\}$ and edge set $E$;
an integer $\Delta\ge3$ bounding the maximum degree.
\Ensure A signing $\sigma:E\to\{-1,1\}$ whose signed adjacency
matrix has operator norm less than $2\sqrt{\Delta-1}$.

\If{$n\le1$}
    \State \Return the empty signing.
\EndIf
\State $\rho_0\gets2\sqrt{\Delta-1}$.
    \Comment{Spectral threshold}
\State $\gamma_0\gets1+1/(64n^2)$.
    \Comment{Progress factor}
\State $C_n\gets \gamma_0$.
\For{$\jmath_0=n-1$ \textbf{downto} $2$}
    \State $C_{\jmath_0}\gets \gamma_0/(1-\gamma_0C_{\jmath_0+1}/4)$.
\EndFor
\State $C_1\gets \gamma_0/(1-\Delta \gamma_0C_2/\rho_0^2)$.
\State $\mathbf{C}\gets(C_1,\ldots,C_n)$.
    \Comment{Weights indexed by recursion depth}
\State $K\gets\varnothing$.
    \Comment{Initialize the active vertex set}
\State $\sigma\gets\varnothing$.
    \Comment{Initialize the empty signing}
\For{$v=1,\ldots,n$}
    \State $(K,\sigma)\gets
    \Call{Fill}{G,K,\sigma,v,\varnothing,\rho_0,\gamma_0,\mathbf{C}}$.
    \Statex \hspace{\algorithmicindent}
    \Comment{Replace the current state by the returned pair}
\EndFor
\State \Return $\sigma$.
    \Comment{All vertices are now active}
\end{algorithmic}
\end{algorithm}

For $n\ge2$, Algorithm~\ref{alg:signing} first computes the
weights from Lemma~\ref{lem:weights} (Lines~6--11, Algorithm~\ref{alg:signing}), initializes
$(K,\sigma)=(\varnothing,\varnothing)$, and then makes the outer
calls in vertex order. Each call to $\Fill$ is required to
return an admissible state whose active set consists precisely
of its input active set together with the vertex it was asked
to insert. The returned state becomes the input state for the
next outer call. Thus, after the outer call that inserts
vertex $v$ is completed, the active set will be
$K=\{1,\ldots,v\}$. 
The proof of Theorem~\ref{thm:main} establishes that every call
terminates and returns an admissible active state $(K,\sigma)$,
where the returned set $K$ is exactly the input set $K$
together with the vertex $v$ to be inserted.

\subsubsection{The recursive subroutine}

A call to $\Fill$ for vertex $v$ has depth $\jmath_0=|P|+1$.
At each iteration, it either selects a neighbor $u_*\in N_G(v)\cap K$
for deletion and recursive reinsertion, or assigns signs to all
edges joining $v$ to vertices in $K$ and then inserts $v$ into $K$.
Algorithm~\ref{alg:fill} specifies the selection rule for $u_*$
and the steps in these two alternatives.

The procedure first computes $U_v=N_G(v)\cap K$ and the values
$\lambda_u$ for $u\in U_v$ from the current state $(K,\sigma)$ (Lines~4--7, Algorithm~\ref{alg:fill}).
If $\jmath_0<n$ and some $u\in U_v$ satisfies
\begin{equation}\label{eq:threshold}
 \lambda_u\ge\gamma_0C_{\jmath_0+1},
\end{equation}
it selects the smallest such vertex $u_*$ (Lines~9--10, Algorithm~\ref{alg:fill}), removes it from $K$,
and erases the signs on its edges to vertices remaining in $K$ (Lines~13--14, Algorithm~\ref{alg:fill}).
It then calls $\Fill$ to restore $u_*$, passing $P\cup\{v\}$
as the set of ancestor vertices.

\begin{algorithm}[H]
\caption{The recursive subroutine $\Fill$}
\label{alg:fill}
\small
\begin{algorithmic}[1]
\Require The graph $G=(V,E)$, with $n=|V|\ge2$, and the fixed
parameters $\rho_0,\gamma_0,\mathbf{C}=(C_1,\ldots,C_n)$
from Algorithm~\ref{alg:signing}.\newline
An admissible state $(K,\sigma)$, a vertex $v\in V\setminus K$,
and an ancestor set $P\subseteq V\setminus(K\cup\{v\})$.
If $P\ne\varnothing$, then $v$ has a neighbor in $P$.
\Ensure An admissible state $(K_{\mathrm{out}},\sigma_{\mathrm{out}})$
with $K_{\mathrm{out}}=K\cup\{v\}$, where $K$ is the input active set;
$\sigma_{\mathrm{out}}$ signs every edge of $G[K_{\mathrm{out}}]$.
\Procedure{Fill}{$G,K,\sigma,v,P,\rho_0,\gamma_0,\mathbf{C}$}
    \State $\jmath_0\gets|P|+1$.
        \Comment{Current recursion depth}
    \Loop
        \State $U_v\gets\{u\in K:\{u,v\}\in E\}$.
            \Comment{Active neighbors of $v$}
        \State Form $\mathbf{A}_K$, the signed adjacency matrix
        of the current state $(K,\sigma)$.
        \State $\mathbf{Q}_K\gets \mathbf{I}_K-\mathbf{A}_K/\rho_0$.
        \State $\lambda_u\gets(\mathbf{Q}_K^{-1})_{uu}$ for each $u\in U_v$.
        \State $u_*\gets\texttt{none}$.
            \Comment{No deletion candidate selected}
        \If{$\jmath_0<n$}
            \State If some $u\in U_v$ satisfies
            $\lambda_u\ge \gamma_0C_{\jmath_0+1}$, set $u_*$ to the
            smallest such vertex.
        \EndIf
        \If{$u_*\ne\texttt{none}$}
            \State Erase from $\sigma$ all signs on
            edges incident to $u_*$.
            \State $K\gets K\setminus\{u_*\}$.
            \State $(K,\sigma)\gets
            \Call{Fill}{G,K,\sigma,u_*,P\cup\{v\},
            \rho_0,\gamma_0,\mathbf{C}}$.
            \State \textbf{continue}.
                \Comment{Recompute from the returned state}
        \Else
            \State Write $U_v=\{u_1<\cdots<u_\ell\}$,
            where $\ell=|U_v|$.
            \State Choose temporary signs
            $\mathbf{s}=(s_{u_i})_{i=1}^{\ell}$ as follows.
            \For{$i=1,\ldots,\ell$}
                \If{$\displaystyle
                    \sum_{k=1}^{i-1}
                    (\mathbf{Q}_K^{-1})_{u_i u_k}s_{u_k}>0$}
                    \State $s_{u_i}\gets-1$.
                \Else
                    \State $s_{u_i}\gets+1$.
                \EndIf
            \EndFor
            \State $K\gets K\cup\{v\}$.
            \State Extend $\sigma$ by setting
            $\sigma(\{v,u_i\})\gets s_{u_i}$
            for $i=1,\ldots,\ell$.
            \State \Return $(K,\sigma)$.
                \Comment{Return the updated active set and signing}
        \EndIf
    \EndLoop
\EndProcedure
\end{algorithmic}
\end{algorithm}

The child call reinserts the deleted vertex $u_*$.
Since $u_*$ is adjacent to the parent's vertex $v$, and
$v$ remains inactive, $u_*$ has an inactive neighbor in
the child's ancestor set $P\cup\{v\}$.
When the child call returns, the parent call for $v$
updates $(K,\sigma)$ to the returned state and starts
the next loop iteration, recomputing the matrix data
and repeating the deletion test
(Line~16, Algorithm~\ref{alg:fill}).

The state $(K,\sigma)$ returned by the child call has the same
vertex set $K$ as before the deletion, but the signing $\sigma$
may have changed. At the start of every iteration, the procedure
therefore recomputes $U_v=N_G(v)\cap K$.
The entries $(\mathbf{Q}_K^{-1})_{uw}$ needed in that iteration
are computed from the current signing $\sigma$ as they are used.
These computations and the subsequent sign choices use the
exact rational arithmetic described in
Proposition~\ref{prop:arithmetic}.

If no vertex is selected for deletion, the procedure uses
Lemma~\ref{lem:signs} to assign a sign in $\{-1,1\}$ to each
edge $\{v,u\}$ with $u\in U_v$ (Lines~20--26, Algorithm~\ref{alg:fill}). It then extends $\sigma$ by
these signs, replaces $K$ with $K\cup\{v\}$, and returns
the updated state $(K,\sigma)$ (Lines~27--29, Algorithm~\ref{alg:fill}). If $U_v=\varnothing$, no signs
need to be assigned.

All deletions and insertions modify only the temporary state
$(K,\sigma)$; the input graph $G=(V,E)$ remains fixed.
Deleting and recursively reinserting $u_*$ may change signs
on edges incident to $u_*$ or to other vertices deleted
during the recursive call. For any edge $\{x,y\}\in E$,
the sign $\sigma(\{x,y\})$ remains unchanged if both $x$
and $y$ remain in $K$ throughout this process.

\subsubsection{Why the deletion is proactive}

The procedure tests~\eqref{eq:threshold} before choosing signs
on the edges from $v$ to $U_v=N_G(v)\cap K$.
It may therefore delete a neighbor $u_*$ even if some choice
of these signs would already allow $v$ to be inserted while
preserving admissibility. This is why we call the deletion
\emph{proactive}. The threshold is chosen to control the
determinant change in both the deletion and insertion cases,
as needed for the termination proof.

By Lemma~\ref{lem:determinants}, deleting a vertex $u_*$
that satisfies~\eqref{eq:threshold} gives
\[
 \frac{D_{K\setminus\{u_*\}}}{D_K}
 =\lambda_{u_*}\ge\gamma_0 C_{\jmath_0+1}.
\]
If $\jmath_0<n$ and no $u\in U_v$ satisfies~\eqref{eq:threshold},
the greedy estimate in Lemma~\ref{lem:signs} and the weight
recurrence instead give a vector $\mathbf{s}$ of signs on
the edges from $v$ to $U_v$ such that
\[
 \frac{D_{K\cup\{v\}}}{D_K}
 =1-q_v(\mathbf{s})\ge\frac{\gamma_0}{C_{\jmath_0}}>0.
\]
Thus the insertion preserves admissibility, with an explicit
lower bound on the determinant ratio.
At depth $\jmath_0=n$, the set $K$ is empty, so no signs are needed
and the same inequality holds with $q_v(\mathbf{s})=0$ and
$C_n=\gamma_0$.
The proof of Theorem~\ref{thm:main} verifies these claims
and uses the two determinant estimates to bound the total
number of vertex insertions and deletions.

To count the vertex operations and specify the states used
in the termination proof, we call each deletion from $K$
or insertion into $K$ an \emph{update}.
A deletion update removes $u_*$ and its incident signs,
and ends when the child call for $u_*$ begins with the
modified state $(K,\sigma)$.
Operations within that child call are counted separately.
An insertion update adds $v$ and its incident signs,
and ends when the call returns and the caller adopts
the returned state.
Starting an outer call or receiving a returned state
does not count as an additional update.

The analysis considers the state $(K,\sigma)$ immediately
after each update and at the start of each outer call.
After a deletion, this is the state passed to the child;
after an insertion, it is the state returned to the caller.
For a completed outer call, the caller is
Algorithm~\ref{alg:signing}.

\subsection{Proof of Theorem~\ref{thm:main}}
\begin{proof}
The proof has three main tasks, organized into five steps.

First, in Steps~1 and~2, we establish the invariants needed
throughout the execution.
Step~1 shows that the vertices associated with unfinished
calls to $\Fill$ are distinct and lie outside the current
vertex set $K$. It also shows that each call for a vertex
$v$, if it returns, returns a state whose vertex set is
its input set together with $\{v\}$.
Step~2 shows that the current state $(K,\sigma)$ is admissible
whenever an outer call begins and immediately after every update.
In particular, $\mathbf{Q}_K$ is positive definite at each of
these points.
This guarantees admissibility on the current set $K$, which
may still be a proper subset of $V$. To obtain a signing of
the whole graph, we must also prove that all calls terminate
and that the final active set is $V$; these conclusions
require Steps~3 and~4.

Second, in Steps~3 and~4, we bound the number of updates
and prove termination and the spectral bound.
Step~3 defines a potential function that decreases at each
update and uses it to show that the total number of vertex
insertions and deletions is less than $130n^3$.
Step~4 uses this update bound to prove that all calls
terminate. It then combines termination with the invariants
from Steps~1 and~2 to obtain the required spectral bound.

Finally, Step~5 combines the update bound with the local
computation cost in Proposition~\ref{prop:arithmetic}
to obtain the total bit-complexity bound.

\smallskip

\noindent\emph{Step 1: the recursion stack and the returned vertex set.}
We show that every call to $\Fill$ for a vertex $v$,
if it returns, returns a state whose vertex set is its input set
$K$ together with $\{v\}$. We also show that the recursion
depth is at most $n$.

The \emph{stack} is the ordered list of unfinished calls to
$\Fill$, from the outer call to the call currently executing.
Let $v_1,\ldots,v_{\jmath_0}$ be the vertices that the unfinished calls
are to insert.
Each $v_i$ is the fourth input argument of the corresponding
call to $\Fill$, as specified in Line~15 of
Algorithm~\ref{alg:signing} for the outer call and Line~15 of
Algorithm~\ref{alg:fill} for each recursive call.
 Let $S$ be the set of vertices for which
Algorithm~\ref{alg:signing} has already begun an outer call
to $\Fill$, including those whose outer calls have finished.

We claim that, for every outer or child call, the following
relations hold at two points: immediately after the call begins
and, if it returns, immediately after it returns.
At each point, $K$ and the stack refer to their states after
the corresponding event:
\begin{equation}\label{eq:SK}
 v_i\ne v_k \quad (1\le i<k\le \jmath_0),
 \qquad
 S=K\mathbin{\dot\cup}\{v_1,\ldots,v_{\jmath_0}\},
\end{equation}
where $\dot\cup$ denotes disjoint union.
At entry, the call has been added to the stack.
Upon return, the call has been removed from the stack,
and the caller has adopted the returned active set $K$.
Thus each vertex in $S$ either belongs to $K$ or is awaiting
insertion by exactly one unfinished call.

We prove these relations (\ref{eq:SK}) by induction over outer-call starts
and updates. Initially, $S=K=\varnothing$ and the stack is empty,
so the relations  (\ref{eq:SK}) hold.
Starting an outer call for $v\notin S$ adds $v$ to $S$
and places the call for $v$ on the stack. Since $v\notin S$,
it belongs neither to $K$ nor to any existing call.
Deleting $u_*\in K$ removes it from $K$ and places the child
call for $u_*$ on the stack. Since the stack vertices lie
outside $K$, this introduces no repeated vertex.
Inserting the current vertex $v_{\jmath_0}$ adds it to $K$
and removes its completed call from the stack.
Neither deletion nor insertion changes $S$.

Each event therefore preserves the distinctness of the stack
vertices and the disjoint union
$S=K\mathbin{\dot\cup}\{v_1,\ldots,v_{\jmath_0}\}$.
This completes the induction and proves~\eqref{eq:SK}.

In particular, every outer call for $v$ begins with $v\notin K$.
During any call to $\Fill$, the vertices in its ancestor set
$P$ remain outside the current set $K$.
Now consider any call to $\Fill$ that returns, whether an
outer call or a child call. Let $v$ denote the vertex that
this call is to insert, and let $K_{\mathrm{in}}$ and
$K_{\mathrm{out}}$ denote the vertex sets of its input
and returned states, respectively.
Let $P$ be the ancestor set supplied to this call.
If it is an outer call, then $P=\varnothing$.
If it is a child call, its ancestor calls remain unfinished
throughout its execution, so their vertex set remains $P$.

No new outer call starts between the entry and return of
the call under consideration, so $S$ is unchanged.
At entry, the unfinished calls correspond to the vertices
in $P\cup\{v\}$. Immediately after this call returns, all
calls it initiated have also returned, leaving precisely
the ancestor calls corresponding to $P$.
The disjoint-union relation therefore gives
\[
 S
 =K_{\mathrm{in}}\mathbin{\dot\cup}P
    \mathbin{\dot\cup}\{v\}
 =K_{\mathrm{out}}\mathbin{\dot\cup}P.
\]
Removing $P$ from both expressions yields
\[
 K_{\mathrm{out}}=K_{\mathrm{in}}\cup\{v\}.
\]
For an outer call, the same argument applies with
$P=\varnothing$.
This proves the required set equality for every call that
returns, without assuming that all calls terminate.

Finally, we use~\eqref{eq:SK} to bound the recursion depth
and the number of active neighbors of the current vertex.
Since $S\subseteq V$, we have
\(
 |K|+\jmath_0=|S|\le n.
\)
Thus the recursion depth is at most $n$, and $K=\varnothing$
whenever $\jmath_0=n$.

Now consider the current vertex $v=v_{\jmath_0}$.
If $\jmath_0=1$, then $v$ has at most $\Delta$ neighbors in $K$.
If $\jmath_0\ge2$, the current call was initiated by the
parent call after deleting its neighbor $v$.
The parent vertex $v_{\jmath_0-1}$ is therefore adjacent to $v$
and remains outside $K$ by~\eqref{eq:SK}.
Consequently, $v$ has at most $\Delta-1$ neighbors in $K$.
Hence
\begin{equation}\label{eq:shang}
 |N_G(v)\cap K|
 \le
 \begin{cases}
 \Delta, & \jmath_0=1,\\
 \Delta-1, & \jmath_0\ge2.
 \end{cases}
\end{equation}
These bounds will be used in Step~2 to show that inserting
the current vertex preserves admissibility.

\smallskip
\noindent\emph{Step 2: every update preserves admissibility.}
We prove that each update preserves positive definiteness of
$\mathbf{Q}_K$ and leaves $\sigma$ defined on exactly the edges
of $G[K]$.

The empty initial state is admissible.
Suppose that the current state $(K,\sigma)$ is admissible.
If $u_*\in K$ is deleted, Lemma~\ref{lem:determinants} gives
\(
 \mathbf{Q}_{K\setminus\{u_*\}}\succ0.
\)

Now suppose that no neighbor of $v$ is selected for deletion.
The procedure then chooses signs on the edges joining $v$
to its neighbors in $K$ and inserts $v$ into $K$.
Recall that $U_v=N_G(v)\cap K$.
Let $\mathbf{s}=(s_u)_{u\in U_v}$ be the sign vector chosen
by the greedy rule in Lemma~\ref{lem:signs}
(Lines~20--26, Algorithm~\ref{alg:fill}), where $s_u$
is the sign assigned to the edge $\{v,u\}$.
If $\jmath_0<n$, the deletion test has failed for every
$u\in U_v$, so
\[
 \lambda_u<\gamma_0C_{\jmath_0+1}
 \qquad(u\in U_v).
\]
Combining Lemma~\ref{lem:signs}, the neighbor bounds
in~\eqref{eq:shang}, and the weight recurrence in Lemma \ref{lem:weights} gives
\begin{equation}\label{eq:safeinsertion0}
 q_v(\mathbf{s})
 \le \frac{1}{\rho_0^2}\sum_{u\in U_v}\lambda_u
 \le 1-\frac{\gamma_0}{C_{\jmath_0}}
 \qquad(1\le\jmath_0<n).
\end{equation}
For $\jmath_0=1$, this uses $|U_v|\le\Delta$ and the
definition of $C_1$. For $\jmath_0\ge2$, it uses
$|U_v|\le\Delta-1$, $\rho_0^2=4(\Delta-1)$, and the
recurrence for $C_{\jmath_0}$.
The estimate also holds for $U_v=\varnothing$, with
$q_v(\mathbf{s})=0$ and an empty sum.

If $\jmath_0=n$, Step~1 gives $K=\varnothing$, so
$q_v(\mathbf{s})=0$ and $C_n=\gamma_0$.
Together with~\eqref{eq:safeinsertion0} and the positivity
of the weights, this proves that every selected insertion
satisfies
\begin{equation}\label{eq:residual}
 1-q_v(\mathbf{s})\ge\frac{\gamma_0}{C_{\jmath_0}}>0.
\end{equation}
Lemma~\ref{lem:determinants} therefore gives
$\mathbf{Q}_{K\cup\{v\}}\succ0$ for the state obtained
with the chosen signs.

It remains to verify that each update leaves $\sigma$ defined
on exactly the edges of the current induced graph $G[K]$.
Deleting $u_*$ erases the signs on all edges incident to $u_*$
and retains every other current sign.
Inserting $v$ assigns signs to all edges joining $v$ to
$U_v=N_G(v)\cap K$, while retaining the signs on the
previously active induced graph.

Together with the positive definiteness proved above,
these observations complete the induction:
after every finite sequence of updates,
$\mathbf{Q}_K\succ0$ and $\sigma$ is defined on exactly
the edges of $G[K]$. No assumption of termination is used.

\smallskip

\smallskip

\noindent\emph{Step 3: a potential for bounding the number of updates.}

In this step, we prove that fewer than $130n^3$ updates can occur by defining
a potential function $\Phi$ that satisfies $\Phi\ge1$
throughout the execution and decreases by a fixed factor
at every update.

We consider the initial state, the state immediately after each
update, and the state at the start of each outer call.
At any such time, let $\kappa$ be the number of outer calls that
have not yet started, and let $\jmath_0$ be the number of unfinished
calls to $\Fill$. Thus $\jmath_0=0$ when no call is unfinished.
Recall that $C_1,\ldots,C_n$ are the positive depth weights
defined in Lemma~\ref{lem:weights}.
For the current state $(K,\sigma)$, define
\begin{equation}\label{eq:potential}
 \Phi:=\frac{C_1^\kappa\prod_{i=1}^{\jmath_0}C_i}{D_K},
\end{equation}
where the empty product is one.
The factor $C_i$ accounts for the unfinished call at depth $i$,
and each outer call not yet started contributes a factor $C_1$.
We call $\Phi$ a \emph{potential function}.
It is used only in the proof and need not be computed by
the algorithm.

Initially, $\kappa=n$, $\jmath_0=0$, and $D_K=1$, so $\Phi=C_1^n$.
When an outer call begins, $\kappa$ decreases by one and $\jmath_0$
changes from $0$ to $1$, while $(K,\sigma)$ remains unchanged.
The factor $C_1$ removed from $C_1^\kappa$ is therefore replaced
by the factor $C_1$ in the product over unfinished calls.
Thus starting an outer call leaves $\Phi$ unchanged.

Consider a deletion update in a call at depth $\jmath_0$,
and let $(K,\sigma)$ be the state immediately before deletion.
By~\eqref{eq:delete} in Lemma~\ref{lem:determinants},
removing $u_*$ from $K$ and retaining the signs on the
remaining edges gives
\[
 D_{K\setminus\{u_*\}}=\lambda_{u_*}D_K,
 \qquad
 \lambda_{u_*}=(\mathbf{Q}_K^{-1})_{u_*u_*},
\]
where $\mathbf{Q}_K$ is evaluated before deletion.
The child call for $u_*$ then begins at depth $\jmath_0+1$,
adding the factor $C_{\jmath_0+1}$ to the numerator
of~\eqref{eq:potential}. The number $\kappa$ and the factors
for all previously unfinished calls remain unchanged.

For each deletion or insertion update, let
$\Phi_{\mathrm{before}}$ and $\Phi_{\mathrm{after}}$ denote
the potential immediately before and after that update,
respectively.
For a deletion, the latter value is taken when the child call
begins, before it performs any updates.
For an insertion, it is taken after the current call returns
and the caller adopts the returned state.

For the deletion update considered above, the determinant
identity and the deletion threshold~\eqref{eq:threshold} give
\begin{equation}\label{eq:deletionprogress}
 \frac{\Phi_{\mathrm{after}}}{\Phi_{\mathrm{before}}}
 =\frac{C_{\jmath_0+1}D_K}{D_{K\setminus\{u_*\}}}
 =\frac{C_{\jmath_0+1}}{\lambda_{u_*}}
 \le\frac{1}{\gamma_0}.
\end{equation}

Now consider an insertion update in a call at depth $\jmath_0$.
According to Lemma \ref{lem:determinants}, inserting $v$ with the chosen signs $\mathbf{s}$ multiplies
the determinant by $1-q_v(\mathbf{s})$.
The call then returns and is removed from the stack,
removing the factor $C_{\jmath_0}$ from the numerator
of the potential, while $\kappa$ and all ancestor factors
remain unchanged.
Thus, by~\eqref{eq:residual},
\begin{equation}\label{eq:insertionprogress}
 \frac{\Phi_{\mathrm{after}}}{\Phi_{\mathrm{before}}}
 =\frac{1}{C_{\jmath_0}(1-q_v(\mathbf{s}))}
 \le\frac{1}{\gamma_0}.
\end{equation}

In~\eqref{eq:deletionprogress}, the child call has just begun;
any updates performed within it are counted separately.
Each ancestor call keeps its factor in~\eqref{eq:potential}
until that call itself finishes.

By Step~2, every state under consideration is admissible.
Lemma~\ref{lem:determinants} therefore gives $0<D_K\le1$,
while Lemma~\ref{lem:weights} gives $C_i\ge1$.
Consequently, $\Phi\ge1$.
Since each update reduces $\Phi$ by at least a factor
$\gamma_0>1$, after any finite number $T$ of updates
performed by the algorithm we have
\begin{equation}\label{eq:Phig0}
 1\le\Phi\le C_1^n\gamma_0^{-T}
       <(5\gamma_0)^n\gamma_0^{-T},
\end{equation}
where the last inequality uses $C_1<5\gamma_0$ from
Lemma~\ref{lem:weights}.

Writing $\log$ for the natural logarithm, we obtain
\[
 \log\gamma_0
 =\int_1^{\gamma_0}\frac{dx}{x}
 \ge\frac{\gamma_0-1}{\gamma_0}
 =\frac{1}{64n^2+1}
 \ge\frac{1}{65n^2}.
\]
Also, $\gamma_0\le65/64$ and $5(65/64)<e^2$, so
$\log(5\gamma_0)<2$.
Taking logarithms in~\eqref{eq:Phig0} and using the bounds
on $\log\gamma_0$ and $\log(5\gamma_0)$ yields
\begin{equation}\label{eq:count}
 T<\frac{n\log(5\gamma_0)}{\log\gamma_0}<130n^3.
\end{equation}
This bound holds for the number of updates performed so far
at any finite stage, without assuming that all calls terminate.

\smallskip

\noindent\emph{Step 4: termination and the spectral bound.}
In this step,
we use the update bound to prove that Algorithm~\ref{alg:signing}
terminates and returns a signing satisfying~\eqref{eq:main}.

Each iteration of the loop in Algorithm~\ref{alg:fill}
performs one update: either it deletes a vertex and invokes
a child call, or it inserts the current vertex and returns.
By Proposition~\ref{prop:arithmetic}, the local computation
preceding each update takes finite time.
Algorithm~\ref{alg:signing} starts at most $n$ outer calls.
An infinite execution would therefore require infinitely
many updates, contradicting~\eqref{eq:count}.
Thus Algorithm~\ref{alg:signing} and every call to $\Fill$
made during its execution terminate.

We may therefore apply the conclusion of Step~1 to all completed
calls and obtain $K=V$ when Algorithm~\ref{alg:signing} finishes.
Step~2 ensures that the final state is admissible.
The spectral constraint~\eqref{eq:spectralconstraint} then yields
the strict bound~\eqref{eq:main}.

\smallskip
\noindent\emph{Step 5: the stated bit complexity.}
By Lemma~\ref{lem:weights}, initialization requires
\(
 O\!(n^3
   \bigl(\log(n+1)+\log(\Delta+1)\bigr)^2)
\)
bit operations. By Proposition~\ref{prop:arithmetic}, the local
computation in each loop iteration requires at most
\(
 O\!(n^7
   \bigl(\log(n+1)+\log(\Delta+1)\bigr)^2)
\)
bit operations. Each iteration performs exactly one update,
and~\eqref{eq:count} bounds the total number of updates by
$130n^3$. Multiplying these bounds, and including the
initialization cost, gives
\[
 O\!\left(n^{10}
   \bigl(\log(n+1)+\log(\Delta+1)\bigr)^2\right)
\]
bit operations.

The remaining operations also fit within this bound.
These include reading the input, finding a bipartition by
a graph traversal when needed, and maintaining vertex labels,
edge signs, and the recursion stack. Even an implementation
that copies the entire current state and ancestor set at
each call needs to copy only $O(n^2\log(n+1))$ bits per call.
This additional cost is absorbed by the bound for the local
computation. The proof is complete.
\end{proof}

\section{Proofs of the auxiliary results}\label{sec:details}

This section proves the auxiliary results stated in
Section~\ref{sec:ingredients}.

\subsection{Determinant identities}

\begin{proof}[Proof of Lemma~\ref{lem:determinants}]
Deleting a vertex gives a principal submatrix of
$\mathbf{Q}_K$, whereas inserting a vertex adds one row and
one column. We use these observations to prove the deletion
and insertion formulas, and then establish the bound
$0<D_K\le1$.

\smallskip
\noindent\emph{Deletion.}
Let $u\in K$. Deleting $u$ leaves the signs on all edges
of $G[K\setminus\{u\}]$ unchanged. Hence the constraint matrix
after deletion is
\[
 \mathbf{Q}_{K\setminus\{u\}}
 =
 \mathbf{Q}_K[K\setminus\{u\},K\setminus\{u\}].
\]
This principal submatrix is positive definite because
$\mathbf{Q}_K\succ0$, so the state after deletion is admissible.
The cofactor formula for the diagonal entries of
$\mathbf{Q}_K^{-1}$ gives
\[
 \frac{D_{K\setminus\{u\}}}{D_K}
 =
 \frac{\det\mathbf{Q}_K[K\setminus\{u\},K\setminus\{u\}]}
      {\det\mathbf{Q}_K}
 =
 (\mathbf{Q}_K^{-1})_{uu}
 =
 \lambda_u.
\]
This proves~\eqref{eq:delete}. When $|K|=1$, the numerator
is the empty determinant, which equals 1 by convention.

\smallskip
\noindent\emph{Insertion.}
Let $v\in V\setminus K$ and set $U_v=N_G(v)\cap K$.
We retain the current signs on $G[K]$ and assign the sign
$s_u$ to each edge $\{v,u\}$ with $u\in U_v$, where
$\mathbf{s}=(s_u)_{u\in U_v}$.
Define $\mathbf{z}\in\R^K$ by
\[
 z_u=
 \begin{cases}
  s_u,&u\in U_v,\\
  0,&u\in K\setminus U_v.
 \end{cases}
\]
Ordering $v$ after the vertices of $K$, the constraint matrix
after insertion is
\[
 \mathbf{Q}_{K\cup\{v\}}
 =
 \begin{pmatrix}
  \mathbf{Q}_K&-\mathbf{z}/\rho_0\\
  -\mathbf{z}^{\mathsf T}/\rho_0&1
 \end{pmatrix}.
\]
We show that both its positive definiteness and its determinant
are determined by the scalar $1-q_v(\mathbf{s})$.

Suppose first that $K\ne\varnothing$.
For $\mathbf{x}\in\R^K$ and $a\in\R$, set
\[
 \mathbf{y}
 =
 \mathbf{x}-\frac{a}{\rho_0}\mathbf{Q}_K^{-1}\mathbf{z}.
\]
By the definition of $q_v$ and the choice of $\mathbf{z}$,
\[
 \frac{1}{\rho_0^2}
 \mathbf{z}^{\mathsf T}\mathbf{Q}_K^{-1}\mathbf{z}
 =q_v(\mathbf{s}).
\]
Expanding $\mathbf{y}^{\mathsf T}\mathbf{Q}_K\mathbf{y}$
using the definition of $\mathbf{y}$ gives
\begin{align*}
 \mathbf{y}^{\mathsf T}\mathbf{Q}_K\mathbf{y}
 &=
 \left(\mathbf{x}-\frac{a}{\rho_0}\mathbf{Q}_K^{-1}\mathbf{z}\right)^{\mathsf T}
 \mathbf{Q}_K
 \left(\mathbf{x}-\frac{a}{\rho_0}\mathbf{Q}_K^{-1}\mathbf{z}\right)\\
 &=
 \mathbf{x}^{\mathsf T}\mathbf{Q}_K\mathbf{x}
 -\frac{2a}{\rho_0}\mathbf{z}^{\mathsf T}\mathbf{x}
 +\frac{a^2}{\rho_0^2}
       \mathbf{z}^{\mathsf T}\mathbf{Q}_K^{-1}\mathbf{z}\\
 &=
 \mathbf{x}^{\mathsf T}\mathbf{Q}_K\mathbf{x}
 -\frac{2a}{\rho_0}\mathbf{z}^{\mathsf T}\mathbf{x}
 +a^2q_v(\mathbf{s}).
\end{align*}
Thus the quadratic form of $\mathbf{Q}_{K\cup\{v\}}$ satisfies
\begin{equation}\label{eq:erci}
 \begin{pmatrix}\mathbf{x}\\a\end{pmatrix}^{\mathsf T}
 \mathbf{Q}_{K\cup\{v\}}
 \begin{pmatrix}\mathbf{x}\\a\end{pmatrix}
 =
 \mathbf{y}^{\mathsf T}\mathbf{Q}_K\mathbf{y}
 +a^2\bigl(1-q_v(\mathbf{s})\bigr).
\end{equation}
The change of variables
$(\mathbf{x},a)\mapsto(\mathbf{y},a)$ is invertible, with inverse
$\mathbf{x}=\mathbf{y}+(a/\rho_0)\mathbf{Q}_K^{-1}\mathbf{z}$.
Since $\mathbf{Q}_K\succ0$, the right-hand side of (\ref{eq:erci}) is positive
for every nonzero $(\mathbf{y},a)$ when $1-q_v(\mathbf{s})>0$.
Conversely, choosing $\mathbf{y}=\mathbf{0}$ and $a=1$
shows that positive definiteness requires $1-q_v(\mathbf{s})>0$.
Therefore insertion preserves admissibility if and only if
$q_v(\mathbf{s})<1$.

To obtain the determinant formula, write the change of variables as
\[
 \begin{pmatrix}\mathbf{y}\\a\end{pmatrix}
 =
 \mathbf{T}\begin{pmatrix}\mathbf{x}\\a\end{pmatrix},
 \qquad
 \mathbf{T}=
 \begin{pmatrix}
  \mathbf{I}_K&-\mathbf{Q}_K^{-1}\mathbf{z}/\rho_0\\
  \mathbf{0}^{\mathsf T}&1
 \end{pmatrix}.
\]
The matrix $\mathbf{T}$ is upper triangular with every diagonal
entry equal to one, so $\det\mathbf{T}=1$.
The quadratic-form identity  (\ref{eq:erci}) is therefore equivalent to
\[
 \mathbf{Q}_{K\cup\{v\}}
 =
 \mathbf{T}^{\mathsf T}
 \begin{pmatrix}
  \mathbf{Q}_K&\mathbf{0}\\
  \mathbf{0}^{\mathsf T}&1-q_v(\mathbf{s})
 \end{pmatrix}
 \mathbf{T}.
\]
Taking determinants and using $\det\mathbf{T}=1$ gives
\[
 D_{K\cup\{v\}}=\det   \mathbf{Q}_{K\cup\{v\}}
 =(\det\mathbf{T})^2D_K\bigl(1-q_v(\mathbf{s})\bigr)
 =D_K\bigl(1-q_v(\mathbf{s})\bigr),
\]
which proves the insertion formula.
If $K=\varnothing$, then $q_v(\mathbf{s})=0$ and the new
constraint matrix is $(1)$, so both conclusions hold directly.

\smallskip
\noindent\emph{The determinant bound.}
Suppose $K\ne\varnothing$, and let
$\mu_1,\ldots,\mu_{|K|}$ be the eigenvalues of $\mathbf{Q}_K$.
Since $\mathbf{Q}_K\succ0$, all $\mu_i$ are positive.
The diagonal entries of $\mathbf{Q}_K$ are all one, so
\[
 \sum_{i=1}^{|K|}\mu_i=\operatorname{tr}\mathbf{Q}_K=|K|.
\]
The arithmetic--geometric mean inequality therefore gives
\[
 0<D_K=\det \mathbf{Q}_K=\prod_{i=1}^{|K|}\mu_i
 \le
 \left(\frac{1}{|K|}\sum_{i=1}^{|K|}\mu_i\right)^{|K|}
 =1.
\]
For $K=\varnothing$, the same bound follows from
$D_\varnothing=1$.
\end{proof}

\subsection{The greedy sign rule}

\begin{proof}[Proof of Lemma~\ref{lem:signs}]
To prove~\eqref{eq:qtrace}, we separate the diagonal and
off-diagonal terms of $q_v(\mathbf{s})$.
The diagonal terms are independent of $\mathbf{s}$.
We group the off-diagonal terms so that each group involves
one newly chosen sign and the preceding signs and
the greedy rule~\eqref{eq:greedy} makes every such group
nonpositive.

Recall that $U_v=N_G(v)\cap K$, and list its vertices in
the fixed order as $u_1,\ldots,u_\ell$, where $\ell=|U_v|$.
By the definition of $q_v$ in~\eqref{eq:qdef},
\[
 q_v(\mathbf{s})
 =
 \frac{1}{\rho_0^2}
 \sum_{i=1}^{\ell}\sum_{k=1}^{\ell}
 (\mathbf{Q}_K^{-1})_{u_i u_k}s_{u_i}s_{u_k}.
\]
Since $s_{u_i}^2=1$ and
$(\mathbf{Q}_K^{-1})_{u_i u_i}=\lambda_{u_i}$
by~\eqref{eq:leverage}, symmetry of $\mathbf{Q}_K^{-1}$ gives
\[
 q_v(\mathbf{s})
 =
 \frac{1}{\rho_0^2}
 \left(
  \sum_{i=1}^{\ell}\lambda_{u_i}
  +2\sum_{i=1}^{\ell}s_{u_i}
    \sum_{k=1}^{i-1}
    (\mathbf{Q}_K^{-1})_{u_i u_k}s_{u_k}
 \right).
\]
When $s_{u_i}$ is chosen, the signs
$s_{u_1},\ldots,s_{u_{i-1}}$ are already fixed.
Rule~\eqref{eq:greedy} therefore ensures that
\[
 s_{u_i}
 \sum_{k=1}^{i-1}
 (\mathbf{Q}_K^{-1})_{u_i u_k}s_{u_k}
 \le0,
\]
since it selects $s_{u_i}=-1$ when the sum is positive and
$s_{u_i}=+1$ otherwise.
Later sign choices do not change this expression.
Consequently, the expansion of $q_v(\mathbf{s})$ above yields
\[
 q_v(\mathbf{s})
 \le \frac{1}{\rho_0^2}\sum_{i=1}^{\ell}\lambda_{u_i},
\]
which is~\eqref{eq:qtrace}.

For the arithmetic bound, assume that the inverse entries
appearing in~\eqref{eq:greedy} are given.
The $\ell$ sums used to choose the signs contain
\(
 \sum_{i=1}^{\ell}(i-1)={\ell(\ell-1)}/2
\)
terms in total, and each sum requires one comparison with zero.
Thus the rule uses $O(\ell^2)$ arithmetic operations and
comparisons.
If $U_v=\varnothing$, no signs are chosen and both sides
of~\eqref{eq:qtrace} are zero.
\end{proof}

\subsection{Positive weights at finite depth}
\begin{proof}[Proof of Lemma~\ref{lem:weights}]
The main step is to bound $C_2,\ldots,C_n$ by an explicit
comparison sequence. This bound also gives a positive lower
bound for the denominator defining $C_1$.
We first derive the asserted weight bounds from the comparison,
then prove the comparison, and finally analyze the exact
computation.

\smallskip
\noindent\emph{Deriving the weight bounds from a comparison.}
Define, within this proof,
\[
 b_h=2-\frac{1}{h+1}=\frac{2h+1}{h+1},
 \qquad 0\le h\le n-2.
\]
We claim that the recurrence defines $C_{n-h}$ and that
\begin{equation}\label{eq:weightcomparison}
 1\le\frac{C_{n-h}}{\gamma_0}\le b_h<2
 \qquad(0\le h\le n-2).
\end{equation}
This claim immediately gives
$\gamma_0\le C_{\jmath_0}<2\gamma_0$ for
$2\le\jmath_0\le n$.

To obtain the bound for $C_1$, recall that
$\gamma_0=1+1/(64n^2)$, so
\begin{equation}\label{eq:gamma2}
 \gamma_0^2
 =1+\frac{1}{32n^2}+\frac{1}{4096n^4}
 \le1+\frac{1}{16n^2}.
\end{equation}
Using $C_2<2\gamma_0$, $\rho_0^2=4(\Delta-1)$, and
$\Delta\ge3$, we obtain
\begin{align*}
 1-\frac{\Delta\gamma_0C_2}{\rho_0^2}
 &>1-\frac{\Delta}{2(\Delta-1)}\gamma_0^2\\
 &\ge1-\frac34\left(1+\frac{1}{16n^2}\right)\\
 &\ge\frac{13}{64}>\frac15.
\end{align*}
The denominator on the left is also at most one.
Thus~\eqref{eq:Croot} defines $C_1$ and gives
$\gamma_0\le C_1<5\gamma_0$.
It remains to prove~\eqref{eq:weightcomparison} to complete
the positivity and size bounds.

\smallskip
\noindent\emph{Proof of the comparison bound.}
We use induction on $h$.
For $h=0$, \eqref{eq:Cterminal} gives
$C_n/\gamma_0=1=b_0$.
Suppose that~\eqref{eq:weightcomparison} holds for
some $0\le h<n-2$.
By~\eqref{eq:Cinternal}, the next weight is
\[
 C_{n-h-1}
 =\frac{\gamma_0}{1-\gamma_0C_{n-h}/4}.
\]
To prove that this weight is well defined and satisfies
the next comparison bound, it suffices to show
\[
 \frac{1}{b_{h+1}}
 <1-\frac{\gamma_0C_{n-h}}4\le1.
\]
Indeed, the definition of $b_h$ gives
\begin{equation}\label{eq:gammabh}
 \frac{4(1-1/b_{h+1})}{b_h}
 =1+\frac{1}{(2h+1)(2h+3)}
 >1+\frac{1}{4n^2}
 >\gamma_0^2.
\end{equation}
Here the first strict inequality follows from $h<n-2$,
and the second follows from 
(\ref{eq:gamma2}).
Combining the induction hypothesis and (\ref{eq:gammabh}), we obtain that
\[
 1-\frac{\gamma_0C_{n-h}}4
 \ge1-\frac{\gamma_0^2b_h}4
 >\frac{1}{b_{h+1}}>0.
\]
Since $C_{n-h}>0$, this denominator is also at most one.
Consequently,
\[
 1\le
 \frac{C_{n-h-1}}{\gamma_0}
 =\frac{1}{1-\gamma_0C_{n-h}/4}
 <b_{h+1}.
\]
This proves~\eqref{eq:weightcomparison} at $h+1$ and
completes the induction. When $n=2$, only the initial
case $h=0$ is needed.

\smallskip
\smallskip
\noindent\emph{Exact computation and bit complexity.}
We compute the weights as unreduced fractions of positive
integers and bound the growth of their numerators and
denominators.

Write $\gamma_0=\alpha/\beta$, where
$\alpha=64n^2+1$ and $\beta=64n^2$.
Starting from $C_n=\alpha/\beta$, if
$C_{\jmath_0+1}=\xi/\zeta$, then
\eqref{eq:Cinternal} and~\eqref{eq:Croot} give
\[
 C_{\jmath_0}=
 \begin{cases}
 \displaystyle
 \frac{4\alpha\zeta}{4\beta\zeta-\alpha\xi},
       &2\le\jmath_0<n,\\[7pt]
 \displaystyle
 \frac{\alpha\rho_0^2\zeta}
      {\beta\rho_0^2\zeta-\Delta\alpha\xi},
       &\jmath_0=1.
 \end{cases}
\]
Here $\xi,\zeta$ are positive integers, and both displayed
denominators are positive by the preceding argument.
Since $\rho_0^2=4(\Delta-1)$ is an integer, each step uses
only a fixed number of integer multiplications and subtractions.

The coefficients multiplying $\xi$ and $\zeta$ have
$O(\log(n+1)+\log(\Delta+1))$ bits.
Thus each step increases the maximum numerator and
denominator bit length by at most this amount.
After at most $n-1$ steps, all stored and intermediate
integers have
\(
 O\!(n\bigl(\log(n+1)+\log(\Delta+1)\bigr))
\)
bits.
Ordinary binary multiplication has quadratic cost in the
maximum operand bit length, and subtraction has linear cost.
Summing over the $n-1$ steps therefore gives
\[
 O\!\left(
 n^3\bigl(\log(n+1)+\log(\Delta+1)\bigr)^2
 \right)
\]
bit operations, as claimed.
\end{proof}

\subsection{Exact arithmetic and the local bit bound}

\begin{proof}[Proof of Proposition~\ref{prop:arithmetic}]
We show how to compute the entries
$(\mathbf{Q}_K^{-1})_{uw}$ for $u,w\in U_v$, where
$U_v=N_G(v)\cap K$, and then use them to perform the deletion
test and the greedy sign selection.
The key observation is that these entries can be obtained
from the inverse of an integer positive definite matrix,
even when $\rho_0$ is irrational.
We first describe this reduction and the resulting local
computations. 
The rational Gaussian elimination argument establishing
the determinant bit bound is given at the end of the proof.

If $U_v=\varnothing$, there are no inverse entries to compute,
no neighbors to test for deletion, and no signs to choose.
We may therefore assume that $U_v\ne\varnothing$.

\smallskip
\noindent\emph{An integer matrix for the required inverse entries.}
Choose a bipartition $V=L\sqcup R$, where $\sqcup$ denotes
disjoint union, and label its two classes so that $v\in R$.
Then every neighbor of $v$ belongs to $L$, and hence
$U_v\subseteq L\cap K$.
Let $\mathbf{B}$ have rows indexed by $L\cap K$ and columns
indexed by $R\cap K$, with entries
\[
 B_{xy}=
 \begin{cases}
  \sigma(\{x,y\}),&\{x,y\}\in E,\\
  0,&\{x,y\}\notin E,
 \end{cases}
 \qquad x\in L\cap K,\quad y\in R\cap K.
\]
Thus every entry of $\mathbf{B}$ belongs to $\{-1,0,1\}$.
Ordering the vertices of $L\cap K$ before those of $R\cap K$
gives
\begin{equation}\label{eq:Qk}
 \mathbf{Q}_K=
 \begin{pmatrix}
  \mathbf{I}_{L\cap K}&-\mathbf{B}/\rho_0\\
  -\mathbf{B}^{\mathsf T}/\rho_0&\mathbf{I}_{R\cap K}
 \end{pmatrix}.
\end{equation}
Since $\mathbf{Q}_K\succ0$, the completion-of-squares argument
used in the proof of Lemma~\ref{lem:determinants} shows that
\[
 \mathbf{I}_{L\cap K}
 -\frac{\mathbf{B}\mathbf{B}^{\mathsf T}}{\rho_0^2}
 \succ0.
\]
Consequently, the matrix
\[
 \mathbf{M}
 =
 \rho_0^2\mathbf{I}_{L\cap K}
 -\mathbf{B}\mathbf{B}^{\mathsf T}
\]
is positive definite. Its entries are integers because
$\rho_0^2=4(\Delta-1)$ is an integer.

The inverse block containing all required entries satisfies
\begin{equation}\label{eq:rationalinverse}
 (\mathbf{Q}_K^{-1})[L\cap K,L\cap K]
 =
 \rho_0^2\mathbf{M}^{-1}.
\end{equation}
To verify this identity, take any
$\mathbf{z}\in\R^{L\cap K}$ and solve
\[
 \mathbf{Q}_K
 \begin{pmatrix}\mathbf{x}\\ \mathbf{y}\end{pmatrix}
 =
 \begin{pmatrix}\mathbf{z}\\ \mathbf{0}\end{pmatrix},
 \qquad
 \mathbf{x}\in\R^{L\cap K},\quad
 \mathbf{y}\in\R^{R\cap K}.
\]
According to (\ref{eq:Qk}), the second block equation gives
$\mathbf{y}=\mathbf{B}^{\mathsf T}\mathbf{x}/\rho_0$.
Substituting this expression into the first block equation
yields
\[
 \frac{\mathbf{M}}{\rho_0^2}\mathbf{x}=\mathbf{z},
 \qquad\text{so}\qquad
 \mathbf{x}=\rho_0^2\mathbf{M}^{-1}\mathbf{z}.
\]
This proves~\eqref{eq:rationalinverse}.
In particular, the inverse entries needed by the algorithm
are rational.

Each entry of $\mathbf{B}\mathbf{B}^{\mathsf T}$ is a sum
of at most $n$ numbers in $\{-1,0,1\}$.
The entries of $\mathbf{M}$ therefore have absolute value
at most $\rho_0^2+n$.
Forming $\mathbf{M}$ requires $O(n^3)$ integer additions
and multiplications. The integers involved, including
intermediate results, have absolute value at most
$\rho_0^2+n=4(\Delta-1)+n$.
Their binary representations therefore require
$O(\log(n+1)+\log(\Delta+1))$ bits each.

\smallskip
\noindent\emph{Recovering the inverse entries from determinants.}
We compute the entries $(\mathbf{M}^{-1})_{ab}$ with
$a,b\in U_v$ by evaluating $O(n^2)$ determinants of integer
matrices. Each entry is represented as a fraction with
an integer numerator and the same positive denominator
$2\det\mathbf{M}$.
 We first give the recovery formulas and their
bit cost.

We use the following computation bound, proved in the final
part of this proof. If $\mathbf{H}$ is a positive definite
integer matrix of order at most $n$ whose entries have
absolute value at most $\rho_0^2+n+1$, then $\det\mathbf{H}$
can be computed in
\begin{equation}\label{eq:detcost}
 O\!\left(n^5
   \bigl(\log(n+1)+\log(\Delta+1)\bigr)^2\right)
\end{equation}
bit operations. The determinant $\det\mathbf{H}$ and all
intermediate integers in this computation have
\(
 O\!\left(n\bigl(\log(n+1)+\log(\Delta+1)\bigr)\right)
\)
bits.

Write $d=\det\mathbf{M}>0$.
For each $a\in U_v$, let $\mathbf{e}_a\in\R^{L\cap K}$
have coordinate one at $a$ and zero elsewhere.
Since $\mathbf{M}$ is invertible and $d=\det\mathbf{M}$,
the matrix determinant lemma
\cite[Example~1.3.24]{HJ13} gives
\begin{equation}\label{eq:rankone}
 \det(\mathbf{M}+\mathbf{w}\mathbf{w}^{\mathsf T})
 =
 d\bigl(1+\mathbf{w}^{\mathsf T}
                 \mathbf{M}^{-1}\mathbf{w}\bigr)
\end{equation}
for every integer vector $\mathbf{w}$ indexed by $L\cap K$.

Define
\[
 r_a=\det(\mathbf{M}+\mathbf{e}_a\mathbf{e}_a^{\mathsf T})-d
 \qquad(a\in U_v).
\]
Taking $\mathbf{w}=\mathbf{e}_a$ in~\eqref{eq:rankone} gives
\[
 (\mathbf{M}^{-1})_{aa}=\frac{r_a}{d}.
\]
For $a,b\in U_v$ with $a<b$ in the fixed vertex order,
taking $\mathbf{w}=\mathbf{e}_a+\mathbf{e}_b$ gives
\[
 (\mathbf{M}^{-1})_{ab}
 =
 \frac{
  \det\!\left(
   \mathbf{M}+(\mathbf{e}_a+\mathbf{e}_b)
              (\mathbf{e}_a+\mathbf{e}_b)^{\mathsf T}
  \right)-d-r_a-r_b
 }{2d}.
\]
Here symmetry of $\mathbf{M}^{-1}$ accounts for the factor
$2$ in the denominator.

These formulas require determinants only of $\mathbf{M}$
and the matrices
\[
 \mathbf{M}+\mathbf{e}_a\mathbf{e}_a^{\mathsf T},
 \qquad
 \mathbf{M}+(\mathbf{e}_a+\mathbf{e}_b)
            (\mathbf{e}_a+\mathbf{e}_b)^{\mathsf T}.
\]
For $a,b\in U_v$ with $a\ne b$, the matrices added to
$\mathbf{M}$ are $\mathbf{e}_a\mathbf{e}_a^{\mathsf T}$
and $(\mathbf{e}_a+\mathbf{e}_b)
(\mathbf{e}_a+\mathbf{e}_b)^{\mathsf T}$.
Both are positive semidefinite, and every entry of either
matrix is zero or one.
 Since $\mathbf{M}\succ0$ and its integer
entries have absolute value at most $\rho_0^2+n$, all
matrices in these determinant computations 
satisfy the hypotheses of the determinant bound~\eqref{eq:detcost}.

We compute one determinant for $\mathbf{M}$, one for each
$a\in U_v$ to recover a diagonal entry, and one for each
unordered pair $\{a,b\}\subseteq U_v$ to recover an off-diagonal
entry. Thus the total is
$1+|U_v|+\binom{|U_v|}{2}=O(n^2)$.
The~\eqref{eq:detcost} gives a total cost of
\[
 O\!\left(n^7
   \bigl(\log(n+1)+\log(\Delta+1)\bigr)^2\right)
\]
bit operations.
The recovery formulas represent every required entry
$(\mathbf{M}^{-1})_{ab}$ over the common denominator $2d$.
By~\eqref{eq:rationalinverse},
\[
 (\mathbf{Q}_K^{-1})_{ab}
 =\rho_0^2(\mathbf{M}^{-1})_{ab}
 \qquad(a,b\in U_v).
\]
Since $\rho_0^2$ is an integer, the same denominator $2d$
also represents all required entries of $\mathbf{Q}_K^{-1}$.
Each numerator is obtained from a fixed number of determinant
additions or subtractions, followed by multiplication by
$\rho_0^2$. The determinant bit bounds therefore give
\[
 O\!\left(n\bigl(\log(n+1)+\log(\Delta+1)\bigr)\right)
\]
bits for every numerator and for $2d$.
The additional arithmetic used to form these ratios is
covered by the total cost above.

\smallskip
\noindent\emph{The deletion test and the greedy sign choices.}
The diagonal entries just computed give
$\lambda_u=(\mathbf{Q}_K^{-1})_{uu}$ for every $u\in U_v$.
When $\jmath_0<n$, Lemma~\ref{lem:weights} supplies a positive
integer numerator and denominator for
$\gamma_0C_{\jmath_0+1}$, with the same asymptotic bit bound.
Thus each test
\[
 \lambda_u\ge\gamma_0C_{\jmath_0+1}
\]
can be performed exactly by cross-multiplication.
Both denominators are positive, so this multiplication
preserves the direction of the inequality.

For the greedy rule, each sum in~\eqref{eq:greedy} consists
of at most $n$ inverse entries, each multiplied by a previously
chosen sign.
Using their common denominator $2d$, we evaluate the sum
by adding or subtracting its integer numerators.
The numerator of the sum requires at most
$O(\log(n+1))$ additional bits.
Because $2d>0$, the sign of the sum is exactly the sign
of this numerator.

All greedy sums together require $O(n^2)$ integer additions
or subtractions, and there are $O(n)$ sign tests and deletion
comparisons.
Their total cost, as well as the cost of forming $\mathbf{M}$,
is bounded by the preceding bit-complexity estimate.
It remains to justify the determinant computation bound
used in this argument.

\smallskip
\noindent\emph{The determinant computation bound.}

We compute $\det\mathbf{H}$ by ordinary Gaussian elimination
over the rational numbers, reducing every fraction to lowest
terms. Since $\mathbf{H}$ is positive definite, all its
leading principal minors are positive, so elimination
requires no row exchanges. The determinant is the product
of the pivots, and the computation uses $O(n^3)$ rational
arithmetic operations. We bound the bit lengths of the
fractions occurring in this computation.

Let $m\le n$ be the order of $\mathbf{H}$. After eliminating
the first $h$ coordinates, write $s_{ij}^{(h)}$ for the
remaining entry at position $(i,j)$, where $i,j>h$.
Set $I_h:=\{1,\ldots,h\}$. The block determinant formula gives
\[
 s_{ij}^{(h)}
 =
 \frac{
  \det\mathbf{H}[I_h\cup\{i\},I_h\cup\{j\}]
 }{
  \det\mathbf{H}[I_h,I_h]
 }.
\]
Here $\mathbf{H}[I,J]$ denotes the submatrix with rows in $I$
and columns in $J$, both in increasing order, and the empty
determinant is one. Indeed, applying the first $h$
elimination steps to the submatrix in the numerator leaves
$s_{ij}^{(h)}$ as its last diagonal entry; the product of
the preceding pivots is $\det\mathbf{H}[I_h,I_h]$.

By Hadamard's inequality, every minor of order $k\le m$
has absolute value at most
\(
 k^{k/2}(\rho_0^2+n+1)^k.
\)
Since $\rho_0^2=4(\Delta-1)$, the numerator and denominator
of every reduced entry therefore have
\(
 O(
 n\bigl(\log(n+1)+\log(\Delta+1)\bigr))
\)
bits. Each elimination update forms products, a quotient,
and a difference of a fixed number of such fractions,
so its temporary integers satisfy the same asymptotic
bit bound. Accumulating the pivot product also satisfies
this bound, since every partial product is a leading
principal minor of $\mathbf{H}$.

Integer multiplication, division, and fraction reduction
by the Euclidean algorithm take at most quadratic time
in the bit lengths of their operands;
see~\cite[Theorems~3.3 and~4.2]{Shoup08}.
Thus each rational arithmetic operation costs
\[
 O\!\left(n^2
   \bigl(\log(n+1)+\log(\Delta+1)\bigr)^2\right)
\]
bit operations. Together with the $O(n^3)$ operation count,
this proves~\eqref{eq:detcost}.

\end{proof}

\section{Conclusion}\label{sec:conc}

We have given a deterministic polynomial-time algorithm that
assigns signs to the edges of any prescribed finite simple
bipartite graph of maximum degree at most an integer $\Delta\ge3$,
so that the signed adjacency matrix has operator norm strictly
less than $2\sqrt{\Delta-1}$.
The graph need not be regular or connected.
The algorithm builds on the recursive repair framework
of Jadbabaie, Saberi, and Sra~\cite{JSS26}, replacing random
choices by deterministic sign-selection and deletion rules.

The bit-complexity bound in Theorem~\ref{thm:main} is not
intended to be optimal; its purpose is to establish deterministic
polynomial running time. In the bit model, the cost of exact
rational arithmetic accounts for the binary lengths of the
integers involved.
In the unit-cost arithmetic model, each exact rational addition,
subtraction, multiplication, division, and comparison has unit
cost, independently of the operand bit lengths.
Under this model, the algorithm admits an $O(n^6)$ implementation.
Indeed, forming the rational matrix described in the proof of
Proposition~\ref{prop:arithmetic} and computing its inverse
by Gaussian elimination takes $O(n^3)$ arithmetic operations
per update. Computing the inverse entries together avoids
separate determinant evaluations.
The remaining work takes $O(n^2)$ operations per update,
and initializing the weights takes $O(n)$ operations.
Combining these costs with the $O(n^3)$ update bound
in~\eqref{eq:count} gives the stated arithmetic bound,
which is also not intended to be optimal.

A natural direction is to develop a similar deterministic
polynomial-time algorithm for prescribed nonbipartite
Ramanujan base graphs.
The first question is whether every connected nonbipartite
$d$-regular Ramanujan graph, with $d\ge3$, admits an edge
signing $\sigma$ satisfying
\(
 \norm{\mathbf{A}_\sigma}\le2\sqrt{d-1}.
\)
Such a signing would produce a Ramanujan $2$-lift of the
prescribed base graph.
The counterexample in~\cite{Xu26} is not Ramanujan and
therefore does not settle this restricted existence question.
An extension of the present method would need to control
both ends of the signed spectrum without relying on
bipartite spectral symmetry.

\vspace{0.5cm}

{\bf Use of artificial intelligence.}
This work was developed through repeated interactions between
the author and ChatGPT. These exchanges contributed to the
design of the algorithm, the development of its proofs, and
the organization and revision of the manuscript.
The author has carefully reviewed all AI-generated material
and verified the mathematical arguments, revising or rewriting
the text as needed. The author takes full responsibility for
the content and correctness of the paper.

\end{document}